\documentclass{article}
\usepackage{amsthm,comment}
\usepackage{natbib} 
\newtheorem{Theorem}{Theorem}

\newtheorem{Lemma}{Lemma}
\newtheorem{Corollary}{Corollary}
\newtheorem{Remark}{Remark}
\newtheorem{Assumption}{Assumption}
\newtheorem{Example}{Example}

\title{\LARGE \bf
Augmenting Automatic Differentiation for a Single-Server Queue via the Leibniz Integral Rule
}

\author{Michael C. Fu \\
Robert H. Smith School of Business \& Institute for Systems Research, \\
University of Maryland, College Park, Maryland 20742, USA; mfu@umd.edu}

\begin{document}
\maketitle
\begin{abstract}
New recursive estimators for computing higher-order derivatives of mean queueing time from a single sample path of a first-come, first-served single-server queue are presented, 
derived using the well-known Lindley equation and applying the Leibniz integral rule of differential calculus.
Illustrative examples are provided. \\

\noindent
{\bf keywords}: queueing theory; Monte Carlo simulation; sample path analysis; derivative estimation; perturbation analysis; likelihood ratio method; score function method, measure-valued differentiation; weak derivatives; Lindley equation; Leibniz rule. 

\end{abstract}
\section{Introduction}

The general derivative estimation problem is to find an unbiased estimator for the output performance measure of a stochastic system with respect to (w.r.t.) some input scalar parameter. 
For a parameter vector, the problem is to find an unbiased stochastic gradient estimator. 
In this work, we consider an output performance measure that is the expectation of a sample performance, as opposed to say a quantile. 
For queueing systems, common sample performances involve waiting times and queue lengths. 

Specifically, derivative estimation for waiting times (customer delays) is analyzed for a first-come, first-served (FCFS) single-server queue -- often indicated using queueing theory notation as a G/G/1 FCFS queue, 
where FCFS is also known as first-in, first-out (FIFO) queue discipline. 
The goal is to estimate higher-order derivatives of the expected waiting time of a customer w.r.t. a parameter in the service time distribution. 
Previous work has provided first-derivative estimators derived using the {\it likelihood ratio method} (LRM) ---   
see \cite{AlSySh68,Gl87,ReWe89,Ru89,Le90,PeFuHuHe18}; 
and {\it perturbation analysis} (PA) ---
see \cite{GoHo87,SuZa88,Su89,Fu89,HoCa91,Gl91,GaShSr92,ZaSu94,FuHu97}.
The chapter by \cite{Fu15a} in a handbook on simulation provides an overview of stochastic gradient estimation, and
both \cite{MoRoFiMn20}
and
\cite{FuHuSc25} 
provide history and perspective on stochastic gradients and their relationship to modern applications in AI/ML such as stochastic gradient descent (SGD) algorithms for training (deep) artificial neural networks and policy gradient algorithms for reinforcement learning.
%
Other recommended textbooks and research monographs addressing stochastic gradiention include
\cite{RuSh93, Pf96, Gl04, AsGl07, CaLa21, VaHe25}.

The main result in this work are new higher-derivative estimators derived using simple calculus, in particular the Leibniz rules for differentiating integrals and products of functions, following the ideas set out in \cite{PuLe22, ReFu24, ReFuLe25}.
Starting with a well-known deterministic recursive equation describing the dynamics of the sample path of the system, the challenge in handling a discontinuity w.r.t. the input parameter is highlighted, viz., 
a ``kink'' in the sample performance at a single point that leads to a jump at the point for the derivative and subsequently all higher-order derivatives. As a result, sample path derivative estimators derived using infinitesimal perturbation analysis (IPA) and corresponding to what automatic differentiation would produce for second- and higher-order derivatives are biased.
Assuming that a single interarrival time is random, 
a simple probabilistic expression is used to express the expected value output performance measure in terms of an integral that can be differentiated using the Leibniz integral rule under appropriate technical conditions. 
This can also be viewed as a version of conditional Monte Carlo, which was previously applied to derive a different but closely related estimator. 
Two illustrative examples for the cases of a location or scale parameter are provided to illustrate the new estimator. 

\section{Problem Setting}

\subsection{Estimation Problem}
Let $W_i, i \in {\cal Z}^+,$ denote the waiting (delay in queue) time of the $i$th customer, 
and $\theta$ denote a parameter in the distribution of the service time of a single customer in the system.
Then the goal is to find unbiased estimators for the $n$th derivative of the expected value of $W_i$ 
w.r.t.\ $\theta, ~n=1, 2,...$, i.e.,  
to find $\widehat{W}^{(n)}_i$ such that
\begin{equation}
E[\widehat{W}^{(n)}_i] = \frac{d^nE[W_i]}{d\theta^n}.
\label{eq:est}
\end{equation}
Notationally, we will also interchangeably use ``prime'' notation for the derivatives , viz., 
$$
W_i' \equiv W^{(1)}_i, W_i'' \equiv W^{(2)}_i, W_i''' \equiv W^{(3)}_i, 
$$
where superscript ``$(n)$'' denotes the $n$th derivative, and
$$
\widehat{W}_i' \equiv \widehat{W}^{(1)}_i, \widehat{W}_i'' \equiv \widehat{W}^{(2)}_i, 
\widehat{W}_i''' \equiv \widehat{W}^{(3)}_i, . 
$$

\subsection{Single-Server Queue Inputs}

The inputs of the single-server queue are $\{A_i, ~i \in {\cal Z}^+\}$ and $\{S_i, ~i \in {\cal Z}^+\}$, 
where $A_i$ is defined as the time between the arrival of the $i$th and $(i$+1)st customer, 
and $S_i$ denotes the service time of the $i$th customer, all assumed to be positive real numbers (later positive random variables). 
Assume for simplicity and without loss of generality, that the first customer arrives at time 0, 
and that the system starts empty, so that the 1st customer has no delay, 
i.e., 
\begin{Assumption} \label{ass1}
$
W_1 \equiv 0.
$
(Hence, $W_1^{(n)} = 0 ~\forall n \geq 1.$)
\end{Assumption}

\noindent
Further assume henceforth that $\theta$ affects the service time of \textbf{\textit{just the first customer}} in a continuous manner,  
i.e., 
\begin{Assumption} \label{ass2}
$S_1(\theta)$ is continuous w.r.t. $\theta$ 
such that $S_1'$ exists almost everywhere (a.e.), and 
$
\frac{dS_i}{d\theta} \equiv S_i' = 0 ~\forall i \geq 2, ~~
\frac{dA_i}{d\theta} \equiv A_i'= 0 ~\forall i.
$
(Hence, for $n$ $\geq$ $1$, $S_i^{(n)} = 0 ~\forall i \geq 2, ~~
A_i^{(n)} = 0 ~\forall i.$)
\end{Assumption}

Note that up to now, no other assumptions have been made on 
$\{A_i \}$ and $\{S_i \}$, in particular regarding probabilistic/statistical properties, 
so that for now they could just be fixed numbers, 
aside from $S_1$, which is a function of $\theta$ but could also be deterministic. 
In other words, these could be viewed as (deterministic) single realizations or possibly a sample path from an underlying stochastic system.

\subsection{FCFS Dynamics via the Lindley Equation}
Under the first-come-first-served (FCFS) service discipline (first-in-first-out FIFO queue discipline), 
the dynamics of the single-server queue are governed by the well-known Lindley equation \citep{Li52}:
\begin{equation}
W_{i+1} =  \left (W_i + S_i - A_i \right)^+ ,~i \geq 1,
\label{eq:Lindley}
\end{equation}
where $x^+ := \max(x,0)$. 
It will be more convenient to express the Lindley equation using indicator functions, i.e., 
\begin{equation}
W_{i+1} =  (W_i + S_i - A_i) {\bf 1}\{W_i + S_i \geq  A_i\},~i \geq 1,
\label{eq:Lindley2}
\end{equation}
where ${\bf 1}\{ \cdot \}$ denotes the indicator function.

In queueing theory terminology, a busy period ends when the indicator function is zero, 
at which point the $i$th customer is the last customer served in that busy period, 
and the $(i+1)$st customer begins a new busy period, in which case  $W_{i+1} =0$.

\subsection{Derivatives}

Straightforward differentiation of Eq.(\ref{eq:Lindley2}) yields
\begin{equation}
W_{i+1}' =  (W_i'+S_i') {\bf 1}\{W_i + S_i >  A_i\},
\label{eq:IPA-Lindley1}
\end{equation}
leading to the following result, using $W_1=0$ and $S_i^{(n)} = 0 ~\forall i \geq 2$.
\begin{Lemma} \label{lemma1}
Under Assumptions \ref{ass1} and \ref{ass2},
\begin{eqnarray}
W_2' &=& S_1' {\bf 1}\{ S_1 >  A_1 \},  \nonumber \\
W_{i+1}' &=&  W_i' {\bf 1}\{W_i + S_i >  A_i\},~i \geq 2,
\label{eq:IPA-Lindley}
\end{eqnarray}
except at the point $W_i + S_i =  A_i$, where they are undefined. 
\end{Lemma}

\begin{proof}
Apply Assumptions \ref{ass1} and \ref{ass2} to (\ref{eq:IPA-Lindley1}) for $i=1$ and $i \geq 2$. 
\end{proof}

\noindent
In queueing theory terminology, when a busy period ends, the derivative resets to zero, 
and in our simplified setting where $\theta$ only affects the first customer's service time, 
remains zero for all subsequent customers, i.e., 
if $i^*$ is the index of the customer ending the first busy period, 
then (\ref{eq:IPA-Lindley}) can be rewritten succinctly as (for $i>1$)
\begin{equation}
W_i' =  S_1' {\bf 1}\{ i \leq i^* \}, \mbox{~where~} i^* \equiv \min \{i\geq 1: W_i + S_i > A_i\}.
\label{eq:BPend}
\end{equation}
In other words, $W_i' =  S_1' $ for all subsequent customers in the 1st customer's busy period (which could be none), and zero for all other subsequent customers. 

Induction on $n$, with Lemma \ref{lemma1} providing the base case $n=1$ gives the next result, 
which is essentially what automatic differentiation would give. 

\begin{Lemma} \label{lemma2}
Under Assumptions \ref{ass1} and \ref{ass2}, 
and assuming furthermore that $S_1^{(n)}$ exists a.e.,
\begin{eqnarray}
W_2^{(n)} &=& S_1^{(n)} {\bf 1}\{ S_1 >  A_1 \}, ~n \geq 1, \nonumber \\
W_{i+1}^{(n)} &=&  W_i^{(n)} {\bf 1}\{W_i + S_i >  A_i\},~i \geq 2,~n \geq 1, \label{eq:IPA-Lindley2} \\
\Longrightarrow W_j^{(n)} &=&  S_1^{(n)} {\bf 1}\{ i \leq i^* \}, ~j \geq 2, ~n \geq 1, \nonumber 
\end{eqnarray}
except at the point $W_i + S_i =  A_i$, where they are undefined. 
\end{Lemma}

\begin{proof}
Straightforward differentiation of Eq.(\ref{eq:Lindley2}) yields
\begin{equation}
W_{i+1}^{(n)} =  (W_i^{(n)}+S_i^{(n)}) {\bf 1}\{W_i + S_i >  A_i\},
\label{eq:IPA-Lindley-n}
\end{equation}
and then apply Assumptions \ref{ass1} and \ref{ass2} to (\ref{eq:IPA-Lindley-n}) for $i=1$ and $i \geq 2$. 
\end{proof}

\noindent
Again, when a busy period ends, the $n$th derivative resets to zero. 

\medskip
\noindent
\textbf{\textit{In the stochastic setting, the first-derivative estimators given by (\ref{eq:IPA-Lindley}) or by (\ref{eq:IPA-Lindley2}) for $n = 1$ will generally be unbiased under mild conditions, 
whereas the second-order and higher derivative estimators given by (\ref{eq:IPA-Lindley2}) for $n \geq 2$ will generally be biased.}}
Henceforth, the estimators given by (\ref{eq:IPA-Lindley}) or by (\ref{eq:IPA-Lindley2}) will be referred to as the {\bf IPA} estimators. 
As noted earlier, when a busy period ends, the IPA estimator resets. 
In the specific setting of Assumption \ref{ass2}, where the parameter affects only the very first service time, the IPA estimator is zero for the rest of the sample path following the first busy period, ended by customer $i^*$. 

\begin{Remark} \label{rem1}
To handle the point where the derivative is not well-defined, 
one approach is to employ impulse/delta functions, but that is not pursued in this work. 
Instead, the problem can be avoided in the stochastic setting by taking all the underlying random variables to be continuously valued with no mass at discrete points, so that the probability of the single point is zero. 
\end{Remark} 

\noindent
We now present the primary assumptions for the stochastic setting. 
\begin{Assumption} \label{ass3}
Customer interarrival times and service times are 
\begin{itemize}
\item[(i)] random variables drawn from (purely) continuous probability distributions, and
\item[(ii)] mutually independent. 
\end{itemize}
\end{Assumption}

\begin{Remark} \label{rem2}
The stochastic setting considered here is more restrictive than needed
-- mainly for ease of exposition --
albeit more general than the usual G/G/1 queue with a renewal arrival process 
and independent and identically distributed (i.i.d.) service times. 
\end{Remark} 

\noindent
The following two examples provide simple illustrative special cases. 

\begin{Example} \label{ex1}
Customers arrive according to a renewal process, and 
service times are mutually independent and independent of the arrival process, 
with the first service time distributed following a uniform distribution $U(\theta-\delta,\theta+\delta),~\delta \leq \theta$, 
i.e., uniformly between the interval with mean $\theta$ and half-width $\delta$. 
In particular, $\theta$ is a location parameter for the distribution, so $S_1'=1$ and $S_1''=0$. 
Hence, the IPA estimators are given by 
\begin{eqnarray}
W_2' &=& {\bf 1}\{ S_1 >  A_1 \},  \nonumber \\
W_{i+1}' &=&  W_i'  {\bf 1}\{W_i + S_i >  A_i\},~i \geq 2, \nonumber \\
W_i^{(n)} &=& 0, ~i \geq 1,~n \geq 2. \nonumber 
\end{eqnarray}
\end{Example}

\begin{Example} \label{ex2}
Customers arrive according to a renewal process, and 
service times are mutually independent and independent of the arrival process, 
with the first service time exponentially distributed with mean $\theta$. 
In particular, $\theta$ is a scale parameter for the distribution, so $S_1'=S_1/\theta$ but again $S_1''=0$. 
Hence, the IPA estimators are given by 
\begin{eqnarray}
W_2' &=& \frac{S_1}{\theta} {\bf 1}\{ S_1 >  A_1 \},  \nonumber \\
W_{i+1}' &=&  W_i' {\bf 1}\{W_i + S_i >  A_i\},~i \geq 2, \nonumber \\
W_i^{(n)} &=& 0, ~i \geq 1,~n \geq 2. \nonumber 
\end{eqnarray}
\end{Example} 

\begin{Example} \label{ex3}
Customers arrive according to a renewal process, and 
service times are mutually independent and independent of the arrival process, 
with the first service time distributed $U(0,\theta)$, so that $\theta$ is again a scale parameter,  
and the IPA estimators are the same as in Example \ref{ex2}.
\end{Example} 

\begin{Remark} \label{rem3}
The expressions in Examples \ref{ex1} and \ref{ex2} hold more generally whenever the parameter is a location or scale parameter, respectively. 
\end{Remark}

\section{Higher-Order Derivative Estimators Derived via the Leibniz Integral Rule}

When $S$ is a continuous random variable, 
its sample derivative can be expressed in terms of its cumulative distribution function (c.d.f.), 
i.e., the following result holds \cite{Gl91, ZaSu94}.

\begin{Lemma}  \label{lemma3}
Under Assumption \ref{ass3}(i),  
\begin{equation}
S' = \left. - \frac{dF(x;\theta)/d\theta}{dF(x;\theta)/dx} \right|_{x=S},~~
S \sim F(\cdot; \theta), 
\end{equation}
where $F$ denotes the c.d.f.\ of $S$.
\end{Lemma}

\noindent
This simplifies to the location and scale parameter cases of Examples \ref{ex1} and \ref{ex2}, respectively, 
which can be shown more generally for discrete random variables, as well, e.g., \cite{Gl91}.

The first two results on the 1st and 2nd derivatives are well known in the literature,
e.g., \cite{Su89, Fu89, HoCa91, Gl91, ZaSu94}, albeit not in this particular form and setting. 
For the 1st derivative, the IPA estimators are unbiased, i.e., Eq.(\ref{eq:est}) holds for $n=1$.
\begin{Theorem} \label{thm1}
Under Assumptions \ref{ass1}-\ref{ass3},
$$
E[W_i'] = \frac{dE[W_i]}{d\theta} ~\forall i \geq 1.
$$
\end{Theorem}

\begin{proof}
Follows straightforwardly using the dominated convergence theorem, 
since the function $W_2(\theta) = (S_1(\theta) - A_1)^+$ is continuous w.r.t $\theta$ under Assumption \ref{ass2} 
for $i=1$, and then by induction for $i \geq 1$, 
with the function $W_{i+1}(\theta) = (W_i(\theta) + S_i - A_1)^+$ continuous w.r.t $\theta$.
Assumptions \ref{ass1} and \ref{ass3} are needed for the quantities to be well-defined. 
\end{proof}

\noindent
However, for the 2nd derivative, the IPA estimators are generally biased, i.e.,  
$$
E[W_j''] \neq \frac{d^2E[W_j]}{d\theta^2}.
$$

\noindent
Examples \ref{ex1} and \ref{ex2} provide counterexamples, since they give $W_j'' = 0 ~\forall j$, 
whereas $\frac{d^2E[W_2]}{d\theta^2}$ may be nonzero, e.g.,  
if $A_1 \sim U(0,1)$ and $S_1 = \theta \in (0,1)$, then $E[W_2] = \theta^2/2$, so $\frac{d^2E[W_2]}{d\theta^2} =1$.  \\
%
In fact, Examples \ref{ex1} and \ref{ex2} can also provide counterexamples for higher-order derivatives, for which all the IPA estimators are 0.

%
%
%
%


\medskip
\noindent
Theorem \ref{thm1} can be used to find an unbiased estimator for the 2nd derivative, starting with  
\begin{Corollary} \label{cor1}
Under Assumptions \ref{ass1}-\ref{ass3},
$$
\frac{dE[W_i']}{d\theta} = \frac{d^2E[W_i]}{d\theta^2} ~\forall i \geq 1.
$$
\end{Corollary}

\noindent
The expectation operators are taken w.r.t. all the randomness in the system, which could include both the arrival and service processes. 
Now consider expectation w.r.t. a single random variable, 
specifically the $i$th interarrival time $A_i \sim g_i$, where $g_i$ is the p.d.f.\ of $A_i$, 
which exists by Assumption \ref{ass3}. 
%
Taking the expectation w.r.t. $A_i$ in (\ref{eq:IPA-Lindley1}), denoted by $E_i$:
$$
E_i[W_{i+1}'] = \int_{0}^{W_i+S_i} (W_i'+S_i') g_i(x) dx = (W_i'+S_i') \int_{0}^{W_i+S_i}  g_i(x) dx,  
$$
where the second equality follows from $W_i'$ and $S_i'$ not depending on $A_i$ (from Assumptions \ref{ass2} and \ref{ass3}).  
To proceed further requires an extension of Assumption \ref{ass2}.  \\ [8pt]
{\bf Assumption \ref{ass2}$'$.}
In addition to Assumption \ref{ass2}, assume 
$S_1(\theta)$ is twice continuous w.r.t. $\theta$ 
such that $S_1''$ exists a.e. \\[8pt]
Differentiating using the Leibniz product and integral rules gives
\begin{equation} \label{eq:Leibniz}
\frac{dE_i[W_{i+1}']}{d\theta} = (W_i''+S_i'') \int_{0}^{W_i+S_i}  g_i(x) dx + (W_i'+S_i')^2 g(W_i+S_i),  
\end{equation} 
so that the following estimator can be shown to be unbiased by Corollary \ref{cor1}:
\begin{equation} \label{eq:Leibniz2}
\widehat{W}_{i+1}'' =  (\widehat{W}_i'' + S_i'') {\bf 1}\{W_i + S_i >  A_i\}  + (W_i'+S_i')^2 g_i(W_i + S_i), 
\end{equation} 
where the 2nd term augments the usual IPA estimator obtained by automatic differentiation. 

\noindent
Considering $i=1$ and $i>1$ separately, using $W_1 = 0$ and $S_i' = 0$ for $i > 1$:
\begin{eqnarray} 
\widehat{W}_2'' &=&  S_1'' {\bf 1}\{ S_1 >  A_1 \}  + (S_1')^2 g_1(S_1),  \label{eq:Leibniz2a} \\
\widehat{W}_{i+1}'' &=&  \widehat{W}_i'' {\bf 1}\{W_i + S_i >  A_i\}  + (W_i')^2 g_i(W_i + S_i),~i > 1, 
\label{eq:Leibniz2b}
\end{eqnarray} 
and the unbiasedness of the second-derivative estimator is easily established.

\begin{Theorem} \label{thm2}
Under Assumptions \ref{ass1}, \ref{ass2}$'$, and \ref{ass3},
$$
E[\widehat{W}_i''] = \frac{d^2E[W_i]}{d\theta^2} ~\forall i \geq 1.
$$
\end{Theorem}

\begin{proof}
Follows directly by applying Corollary \ref{cor1}, 
since the expectation of (\ref{eq:Leibniz2}) equals (\ref{eq:Leibniz}), 
using induction on $i>1$ for (\ref{eq:Leibniz2b}) starting with $i=1$ for (\ref{eq:Leibniz2a}).  \\
Starting with $i=1$: 
\begin{eqnarray*} 
E[\widehat{W}_2''] &=&  E[S_1'' {\bf 1}\{ S_1 >  A_1 \}  + (S_1')^2 g_1(S_1)]  \\
 &=& E\left[ \int_{0}^{S_1} S_1'' g_1(x) dx + (S_1')^2 g_1(S_1) \right]  \\
 &=& E\left[ \frac{d}{d\theta}  \int_{0}^{S_1} S_1' g_1(x) dx \right]  
 = \frac{d}{d\theta} E\left[  \int_{0}^{S_1} S_1' g_1(x) dx \right]  \\
 &=& \frac{d}{d\theta} E\left[ S_1'  {\bf 1}\{ S_1 >  A_1 \} \right]  
 =  \frac{dE[W_2']}{d\theta} \\
 &=&  \frac{d^2E[W_2]}{d\theta^2} \mbox{~by Corollary \ref{cor1}}.  
\end{eqnarray*} 
General $i$ follows inductively in the same manner: 
\begin{eqnarray*} 
E[\widehat{W}_{i+1}''] &=&  E[(W_i''+S_i'') {\bf 1}\{ W_i+S_i >  A_i \}  + (W_i'+S_i')^2 g_i(W_i+S_i)]  \\
 &=& E\left[ \int_{0}^{W_i+S_i} (W_i''+S_i'') g_i(x) dx + (W_i'+S_i')^2 g_i(W_i+S_i) \right]  \\
 &=& E\left[ \frac{d}{d\theta}  \int_{0}^{W_i+S_i} (W_i'+S_i') g_i(x) dx \right]  
 = \frac{d}{d\theta} E\left[  \int_{0}^{W_i+S_i} (W_i'+S_i') g_i(x) dx \right]  \\
 &=& \frac{d}{d\theta} E\left[ (W_i'+S_i') {\bf 1}\{ W_i+S_i >  A_i \} \right]  
 =  \frac{dE[W_{i+1}']}{d\theta}  \\
 &=&  \frac{d^2E[W_{i+1}]}{d\theta^2} \mbox{~by Corollary \ref{cor1}}.  
\end{eqnarray*} 
Note that the dominated convergence theorem applies throughout where needed, 
because in addition to Assumptions \ref{ass1} and \ref{ass2}, 
$g_i$ is a p.d.f. (existence due to Assumption \ref{ass3}), hence automatically integrable. 
\end{proof}

\medskip
\noindent
Theorem \ref{thm2} can be used to find an unbiased estimator for the 3rd derivative, starting with the analogous result to Corollary \ref{cor1}.  
\begin{Corollary} \label{cor2}
Under Assumptions  \ref{ass1}, \ref{ass2}$'$, and \ref{ass3},
$$
\frac{dE[\widehat{W}_i'']}{d\theta} = \frac{d^3E[W_i]}{d\theta^3} ~\forall i \geq 1.
$$
\end{Corollary}

\medskip
\noindent
The third-derivative estimator requires a corresponding extension of Assumption \ref{ass2}$'$:  \\ [8pt]
{\bf Assumption \ref{ass2}$''$}.
In addition to Assumption \ref{ass2}', assume 
$S_1(\theta)$ is thrice continuous w.r.t. $\theta$ 
such that $S_1'''$ exists almost everywhere (a.e.). 

\medskip
\noindent
Now take the expectation of the {\bf 1st term} of (\ref{eq:Leibniz2})
-- which corresponds to the IPA estimator --- again w.r.t. $A_i$:
$$
E_i[W_{i+1}''] = \int_{0}^{W_i+S_i} (\widehat{W}_i''+S_i'') g_i(x) dx 
= (\widehat{W}_i''+S_i'') \int_{0}^{W_i+S_i}  g_i(x) dx,  
$$
where the second equality follows from $\widehat{W}_i''$ and $S_i''$ not depending on $A_i$.

Differentiating using the Leibniz product and integral rules gives
$$
\frac{dE_i[W_{i+1}'']}{d\theta} = 
(\widehat{W}_i'''+S_i''') \int_{0}^{W_i+S_i}  g_i(x) dx + (\widehat{W}_i''+S_i'')(W_i'+S_i') g_i(W_i+S_i),  
$$
and adding that to the derivative w.r.t. $\theta$ of the 2nd term of (\ref{eq:Leibniz2}) given by  
$$
\frac{d}{d\theta} \left[ (W_i'+S_i')^2 g_i(W_i + S_i) \right] = 
2(W_i'+S_i') (\widehat{W}_i''+S_i'') g_i(W_i + S_i) + (W_i' + S_i')^3 g_i'(W_i + S_i),
$$
yields
\begin{eqnarray} \nonumber
\frac{dE_i[\widehat{W}_{i+1}'']}{d\theta} &=& 
(\widehat{W}_i'''+S_i''') \int_{0}^{W_i+S_i}  g_i(x) dx + 3(\widehat{W}_i''+S_i'')(W_i'+S_i') g_i(W_i+S_i) \\
&& + (W_i' + S_i')^3 g_i'(W_i + S_i),
\label{eq:Leibniz3rd} 
\end{eqnarray} 
which leads to the following estimator that can be shown to be unbiased under additional conditions on the interarrival p.d.f.s $\{g_i\}$:
\begin{eqnarray} 
\widehat{W}_{i+1}''' &=& (\widehat{W}_i''' + S_i''')  {\bf 1}\{ W_i + S_i >  A_i \} 
+ 3(W_i' +S_i') (\widehat{W}_i'' + S_i'') g_i(W_i + S_i) \nonumber \\ 
&& + (W_i' + S_i')^3 g_i'(W_i + S_i).
\label{eq:Leibniz3}
\end{eqnarray} 

\noindent
Again, considering $i=1$ and $i>1$ separately, using $W_1 = 0$ and $S_i' = 0$ for $i > 1$:
\begin{eqnarray} 
\widehat{W}_2''' & \hspace*{-7pt} =& \hspace*{-7pt} S_1'''  {\bf 1}\{ S_1 >  A_1 \} + 3 S_1' S_1'' g_1(S_1)
+ (S_1')^3 g_1'(S_1),  \label{eq:Leibniz3a} \\
\widehat{W}_{i+1}'''  & \hspace*{-7pt} =& \hspace*{-7pt}  \widehat{W}_i'''  {\bf 1}\{ W_i + S_i >  A_i \} + 3W_i' \widehat{W}_i'' g_i(W_i + S_i) + (W_i')^3 g_i'(W_i + S_i),~i > 1.~~~~~~~
\label{eq:Leibniz3b}
\end{eqnarray} 
For the estimators (\ref{eq:Leibniz3}) or (\ref{eq:Leibniz3a})/(\ref{eq:Leibniz3b})
to be well-defined, an additional technical assumption is imposed, 
under which an analogous result to Corollary \ref{cor1} from Theorem \ref{thm2} can be stated.
\begin{Assumption} \label{ass4}
The interarrival time p.d.f.s $\{ g_i \}$ are a.s. continuous on the positive real line.  \\
\end{Assumption}

\begin{Remark} \label{rem4}
This simple sufficient condition allows application of the dominated convergence theorem, but
much weaker conditions can be stated in terms of the support of the individual interarrival and service time distributions. 
For example, it applies to exponentially distribution interarrival times (Poisson arrival process) but excludes uniformly distributed interarrival times. 
\end{Remark}

\begin{Theorem} \label{thm3}
Under Assumptions \ref{ass1}, \ref{ass2}$''$, \ref{ass3}, and \ref{ass4},
$$
E[\widehat{W}_i'''] = \frac{d^3E[W_i]}{d\theta^3} ~\forall i \geq 1.
$$
\end{Theorem}

\begin{proof}
Assumption \ref{ass4} enables interchange of expectation and differentiation to establish
$$
E\left[ \frac{d}{d\theta} \left[ (W_i'+S_i')^2 g_i(W_i + S_i) \right] \right] = 
\frac{d}{d\theta} \left[ E[(W_i'+S_i')^2 g_i(W_i + S_i)] \right],
$$
which may not hold if $g$ is discontinuous on the support of its domain.  \\ 

\noindent
The rest of the proof proceeds analogously to the proof of Theorem \ref{thm2}
by applying Corollary \ref{cor2}, 
since the expectation of (\ref{eq:Leibniz3}) equals (\ref{eq:Leibniz3rd}). 
However, the proof is cleaner if the two parts comprising the final estimator are treated separately. \\
Starting with $i=1$: 
$
\widehat{W}_2''' =  [S_1''' {\bf 1}\{ S_1 >  A_1 \}  +  S_1' S_1'' g_1(S_1)]
 + [2S_1' S_1'' g_1(S_1) + (S_1')^3 g_1'(S_1)],
$
so 
$$
E[\widehat{W}_2'''] =  E[S_1''' {\bf 1}\{ S_1 >  A_1 \}  +  S_1' S_1'' g_1(S_1)]   
+E[2S_1' S_1'' g_1(S_1) + (S_1')^3 g_1'(S_1)] ,
$$
considered separately: 
\begin{eqnarray*} 
E[S_1''' {\bf 1}\{ S_1 >  A_1 \}  +  S_1' S_1'' g_1(S_1)]  
 &=& E\left[ \frac{d}{d\theta}  \int_{0}^{S_1} S_1'' g_1(x) dx \right] \\
 &=& \frac{d}{d\theta} E\left[  \int_{0}^{S_1} S_1'' g_1(x) dx \right]  \\
 &=& \frac{d}{d\theta} E\left[ S_1''  {\bf 1}\{ S_1 >  A_1 \} \right].  \\
%
E[ 2S_1' S_1'' g_1(S_1) + (S_1')^3 g_1'(S_1)]
 &=& E\left[ \frac{d}{d\theta}   (S_1')^2 g_1(S_1) \right]  \\
 &=& \frac{d}{d\theta} E\left[   (S_1')^2 g_1(S_1) \right], 
\end{eqnarray*} 
where the last equality follows from the dominated convergence theorem under Assumption \ref{ass4}. 
Putting the two parts together, 
$$
E[\widehat{W}_2''']  
=  \frac{d}{d\theta} E \left[ S_1''  {\bf 1}\{ S_1 >  A_1 \} + (S_1')^2 g_1(S_1) \right]
=  \frac{dE[W_2'']}{d\theta} 
 =  \frac{d^3E[W_2]}{d\theta^3}  \mbox{~by Corollary \ref{cor2}}.  
$$
\noindent
The case for $i > 2$ follows inductively in the same manner for the decomposed estimator:
$$
\widehat{W}_i''' =  [\widehat{W}_i''' {\bf 1}\{ W_i+S_i >  A_i \}  +  W_i' \widehat{W}_i'' g_i(W_i+S_i)]   
 +[2 W_i' \widehat{W}_i'' g_i(W_i+S_i) + (W_i')^3 g_i'(W_i+S_i)] , 
$$
\begin{eqnarray*} 
E[\widehat{W}_i''' {\bf 1}\{ W_i+S_i >  A_i \}  +  W_i' \widehat{W}_i'' g_i(W_i+S_i)]  
 &=& E\left[ \frac{d}{d\theta}  \int_{0}^{W_i+S_i} W_i'' g_i(x) dx \right] \\
 &=& \frac{d}{d\theta} E\left[  \int_{0}^{W_i+S_i} W_i'' g_i(x) dx \right]  \\
 &=& \frac{d}{d\theta} E\left[ W_i''  {\bf 1}\{ W_i+S_i >  A_1 \} \right].  \\
E[ 2W_1' \widehat{W}_i''  g_i(W_i+S_i) + (W_i')^3 g_i'(W_i+S_i)]
 &=& E\left[ \frac{d}{d\theta}   (W_i')^2 g_i(W_i+S_i) \right]  \\
 &=& \frac{d}{d\theta} E\left[   (W_i')^2 g_i(W_i+S_i) \right], 
\end{eqnarray*} 
where the last equality follows from the dominated convergence theorem under Assumption \ref{ass4}. 
Putting the two parts together, 
$$
E[\widehat{W}_i''']  
=  \frac{d}{d\theta} E \left[ W_i''  {\bf 1}\{ W_i+S_i >  A_i \} + (W_i')^2 g_i(W_i+S_i) \right]
=  \frac{dE[W_i'']}{d\theta} 
 =  \frac{d^3E[W_i]}{d\theta^3}  
 $$
by Corollary \ref{cor2}.  
\end{proof}

\medskip
\noindent
Applying these results to Examples \ref{ex1} and \ref{ex2} simplifies the expressions for the higher-order derivatives, since $S_i'' = 0 ~\forall i$.

\begin{eqnarray}
\widehat{W}_2'' &=&  (S_1')^2 g_1(S_1),  \label{eq:ex1} \\
\widehat{W}_{i+1}'' &=&  \widehat{W}_i'' {\bf 1}\{W_i + S_i >  A_i\}  + (W_i')^2 g_i(W_i + S_i),~i > 1, \label{eq:ex2} \\
\widehat{W}_2''' &=& (S_1')^3 g_1'(W_1 + S_1),  \label{eq:ex3} \\
\widehat{W}_{i+1}'''  &=& \widehat{W}_i'''  {\bf 1}\{ W_i + S_i >  A_i \} + 3W_i' \widehat{W}_i'' g_i(W_i + S_i) + (W_i')^3 g_i'(W_i + S_i),~i > 1. \label{eq:ex4}
\end{eqnarray}

\begin{Example} \label{ex4}
G/D*/1 Queue: 
Customers arrive according to a renewal process, and 
service times are deterministic (but possibly all different) with $S_1 = \theta$. 
In this case, the higher-order estimators for $W_2$ are deterministic and identically equal to the desired derivatives, i.e., estimators (\ref{eq:ex1}) and (\ref{eq:ex3}) become 
$$
\widehat{W}_2'' = g_1(\theta) = \frac{d^2E[W_2]}{d\theta^2}, 
\widehat{W}_2''' = g_1'(\theta) = \frac{d^3E[W_2]}{d\theta^3},
$$
whereas the estimators (\ref{eq:ex2}) and (\ref{eq:ex4}) for subsequent customers remain stochastic (assuming subsequent interarrrival times are random variables) until they vanish at the end of a busy period. 

\end{Example} 
\noindent
Note that if the first interarrival time is uniformly distributed, then $g_1'=0$ a.e., 
but then Assumption \ref{ass2}' is not satisfied. 

\begin{Remark} \label{rem5}
Strictly speaking, Theorems \ref{thm1}-\ref{thm3} do not apply to Example \ref{ex4}, since Assumption \ref{ass3} is not satisfied, but for deterministic services times, it can be shown (but not done here) that all three theorems do in fact apply, just as stated in Remark \ref{rem3} regarding Examples \ref{ex1} and \ref{ex2} holding more generally whenever the parameter is a location or scale parameter, respectively. 
\end{Remark} 

\begin{Example} 
M/G/1 Queue: 
Customers arrive according to a Poisson process with rate $\lambda$, 
and $S_1 \sim F(\cdot; \theta)$, for $F$ satisfying Assumption \ref{ass2}'', independent of the arrival process,
so Theorem \ref{thm3} applies, with the p.d.f.\ for the i.i.d.\ interarrival times 
given by $g_i(x) \equiv g(x)$, viz., 
\begin{eqnarray*}
g(x) &=& \lambda e^{-\lambda x},  \\
g'(x) &=& -\lambda^2 e^{-\lambda x},  \\
g^{(n)}(x) &=& (-1)^n(\lambda)^{n+1} e^{-\lambda x}.  
\end{eqnarray*}
In this case, recalling that $i^* \equiv \min \{i\geq 1: W_i + S_i > A_i; W_1=0 \}$ denotes the index of the customer ending the first busy period, the estimators are given by 
\begin{eqnarray*} 
W_i' &=& S_1' {\bf 1}\{  S_1  >  A_1 \} {\bf 1}\{ i \leq i^* \},~i > 1,  \\
\widehat{W}_2'' &=&  S_1'' {\bf 1}\{ S_1 >  A_1 \}  + (S_1')^2 \lambda e^{-\lambda S_1},   \\
\widehat{W}_{i+1}'' &=&  \widehat{W}_i'' {\bf 1}\{W_i + S_i >  A_i\}  + (W_i')^2 \lambda e^{-\lambda(W_i + S_i)},~i > 1, \\
\widehat{W}_2''' &=& S_1'''  {\bf 1}\{ S_1 >  A_1 \} + 3 S_1' S_1'' \lambda e^{-\lambda S_1} - (S_1')^3 \lambda^2 e^{-\lambda S_1},   \\
\widehat{W}_{i+1}'''  &=& \widehat{W}_i'''  {\bf 1}\{ W_i + S_i >  A_i \} + 3W_i' \widehat{W}_i'' \lambda e^{-\lambda(W_i + S_i)} - (W_i')^3 \lambda^2 e^{-\lambda(W_i + S_i)},~i > 1. 
\end{eqnarray*} 
\end{Example}

\section{Conclusions and Future Research}

By applying the Leibniz calculus rules of differentiation of integrals and products of functions, 
new estimators for higher-order derivatives of customer waiting times in the FCFS single-server queue have been derived. 
Prior to this, higher-order derivative estimators have been derived using likelihood ratio methods and second-derivative estimators using conditional Monte Carlo, but the latter had not been extended to higher-order derivatives. 
The second-derivative estimator derived here differs from the latter but can be transformed to it using further conditioning. 

The derivative estimators derived here are for a single customer's waiting time w.r.t. a parameter in the 1st service time distribution for derivatives up through the 3rd order. 
However, ongoing and future work includes extending and generalizing the results in numerous directions, including the following:
\begin{itemize}
\item[(i)] parameter as a common parameter for all customer service times, e.g., the independent and identically distributed (i.i.d.) case; 
\item[(ii)] customer total time in system;
\item[(iii)] average customer waiting times; 
\item[(iv)] parameters in the interarrival time distribution; 
\item[(v)] general $n$th-order derivatives.  
\end{itemize}
Each of these could be addressed
based on the general framework for applying the Leibniz rule developed in \cite{ReFuLe25}.
Also of practical interest would be to compare the statistical properties of the new higher-order derivative estimators presented here with other estimators derived using LRM, weak derivatives \citep{Pf89}, augmented IPA \citep{GaShSr92}, smoothed perturbation analysis \citep{GoHo87,FuHu97}, and other techniques, both theoretically and empirically. 
Preliminary numerical results for 2nd-derivative estimation of expected system time reported in \cite{Fu26b} are promising.

\bibliographystyle{Chicago}
\bibliography{GG1}

\end{document}